\documentclass[12pt]{article}
\usepackage{fullpage,hyperref,bbm,amssymb,amsfonts,amsthm}
\usepackage[matrix,frame,arrow]{xy}

\input{Qcircuit}

\newcommand{\<}{\langle}
\renewcommand{\>}{\rangle}
\newcommand{\C}{\mathbb{C}}
\newcommand{\cA}{\mathcal{A}}
\newcommand{\cB}{\mathcal{B}}

\newcommand{\cM}{\mathcal{M}}

\newcommand{\E}{\mathrm{\mathbf{E}}}
\newcommand{\Var}{\mathrm{\mathbf{Var}}}
\newcommand{\EE}[1]{\E\left(#1\right)}
\newcommand{\VV}[1]{\Var\left(#1\right)}

\newcommand{\R}{\mathbb{R}}
\renewcommand{\ket}[1]{\left| #1\right\rangle}      
\renewcommand{\bra}[1]{\left\langle #1\right|}      
\newcommand{\ii}{\mathbb{I}}
\newcommand{\ep}{\epsilon}
\newtheorem{definition}{Definition}
\newtheorem{theorem}{Theorem}
\newtheorem{lemma}{Lemma}

\begin{document}
	\title{Quantum Speed-up for Approximating Partition Functions}
		\author{Pawel Wocjan\thanks{School of Electrical Engineering and Computer
		Science,
		University of Central Florida, Orlando, FL~32816, USA. Email:
		\texttt{wocjan@eecs.ucf.edu}}
		\quad	
		Chen-Fu Chiang\thanks{School of Electrical Engineering
		and Computer Science,
		University of Central Florida, Orlando, FL~32816, USA. Email:
		\texttt{cchiang@eecs.ucf.edu}}
		\quad	
		Anura Abeyesinghe\thanks{School of Electrical Engineering
		and Computer Science, University of Central Florida, Orlando, FL~32816, USA.}
		\quad	
		Daniel Nagaj		\thanks{Research Center for Quantum Information, Institute of Physics, Slovak
		Academy of Sciences, D\'ubravsk\'a cesta 9, 84215 Bratislava, Slovakia, and
		Quniverse, L\'{i}\v{s}\v{c}ie \'{u}dolie 116, 84104, Bratislava, Slovakia.
		Email: \texttt{daniel.nagaj@savba.sk}}}
	\maketitle
	
\begin{abstract}
We achieve a quantum speed-up of fully polynomial randomized approximation schemes (FPRAS) for estimating partition functions that combine simulated annealing with the Monte-Carlo Markov Chain method and use non-adaptive cooling schedules.  The improvement in time complexity is twofold:  a quadratic reduction with respect to the spectral gap of the underlying Markov chains and a quadratic reduction with respect to the parameter characterizing the desired accuracy of the estimate output by the FPRAS.  Both reductions are intimately related and cannot be achieved separately.

First, we use Grover's fixed point search, quantum walks and phase estimation to efficiently prepare approximate coherent encodings of stationary distributions of the Markov chains. The speed-up we obtain in this way is due to the quadratic relation between the spectral and phase gaps of classical and quantum walks.  Second, we generalize the method of quantum counting, showing how to estimate expected values of quantum observables.  Using this method instead of classical sampling, we obtain the speed-up with respect to accuracy.  

\end{abstract}

\section{Introduction}
Quantization of classical Markov chains has been crucial in the design of efficient quantum algorithms for a wide range of search problems that outperform their classical counterparts.  We refer the reader to the survey article \cite{Santha} for a detailed account of the rapidly growing collection of quantum-walk-based search algorithms. In this context, we also point to the work \cite{Somma2}, where the authors apply quantized Markov chains to speed up search algorithms based on simulated annealing for finding low-energy states of (classical) Hamiltonians.

In this paper, we extend the scope of use of quantized Markov chains beyond search problems. We show how to employ them to speed up fully polynomial-time randomized approximation schemes for partition functions, based on simulated annealing and the Monte Carlo Markov Chain (MCMC) method. To achieve this improvement, we rely on Szegedy's general method to quantize classical Markov chains \cite{Szegedy,Magniez}, which we review in Appendix \ref{walksection}. 
This method gives us a unitary {\em quantum walk} operator $W(P)$ corresponding to one update step of the classical Markov chain $P$.
The complexity of the classical algorithms we are speeding up is measured in the number of Markov chain invocations. Similarly, 
we express the complexity of our quantum algorithm as the number of times we have to apply a quantum walk operator.
As shown in \cite{efficientWALK}, this operator can be implemented precisely and efficiently.

Sampling from stationary distributions of Markov chains combined with simulated
annealing is at the heart of many clever classical approximation algorithms.
Notable examples include the algorithm for approximating the volume of convex
bodies \cite{Vempala}, the permanent of a non-negative matrix \cite{Vigoda},
and the partition function of statistical physics models such as the Ising
model \cite{Jerrum2} and the Potts model \cite{Vazirani}. Each of these algorithms is a {\em fully
polynomial randomized approximation scheme} (FPRAS), outputting a random number
$\hat{Z}$ within a factor of $(1\pm \ep)$ of the real value $Z$, with probability 
greater than $\frac{3}{4}$, i.e.
\begin{equation}
	\Pr\big[ (1-\epsilon)Z \le \hat{Z} \le (1+\epsilon)Z \big] \ge \frac{3}{4},
	\label{fpras}
\end{equation}
in a number of steps polynomial in $1/\ep$ and the problem size.
		
We show how to use a quantum computer to speed up a class of FPRAS 
for estimating partition functions that rely on simulated annealing 
and the Monte Carlo Markov Chain method (e.g. \cite{Jerrum2,Vazirani}).
Let us start with an outline of these classical algorithms.
Consider a physical system with state space $\Omega$ and an energy function
$E: \Omega\rightarrow \R$, assigning each state $\sigma\in\Omega$
an energy $E(\sigma)$.	
The task is to estimate the Gibbs partition function
		\begin{eqnarray}
		         Z(T) = \sum_{\sigma \in \Omega} e^{-\frac{E(\sigma)}{kT}}
		         \label{partition}
		\end{eqnarray}
		at a desired (usually low) temperature $T_F$. The partition function $Z(T)$ encodes the 
		thermodynamical properties of the system in equilibrium at temperature $T$,
		where the probability of finding the system in state $\sigma$ is given by the Boltzmann distribution
		\begin{eqnarray}
		         \pi_i(\sigma) = \frac{1}{Z(T)}\, e^{-\frac{E(\sigma)}{kT}}.
		         \label{boltzmann}
		\end{eqnarray}
It is hard to estimate $Z(T)$ directly. The schemes we want to speed up
use the following trick. Consider a sequence of decreasing temperatures 
		$
		T_0 \geq T_1 \geq \dots \geq T_{\ell},
		$
		where $T_0$ is a very high starting temperature and $T_{\ell}=T_F$ is the desired
		final temperature. Then, $Z(T_F)$ can be expressed as a telescoping product
		\begin{equation}
			Z(T_F) = Z_0 \, \frac{Z_1}{Z_0} \cdots \frac{Z_{\ell-1}}{Z_{\ell-2}} \frac{Z_\ell}{Z_{\ell-1}} 
			= 
			Z_0 \underbrace{ \left(\alpha_0 \alpha_1 \cdots \alpha_{\ell-2} \alpha_{\ell-1} \right)}_{\alpha}\,,
			\label{telescope}
		\end{equation}
		where $Z_i = Z(T_i)$ stands for the Gibbs partition function at temperature
		$T_i$ and $\alpha_i = Z_{i+1}/Z_{i}$. 
		It is easy to calculate the partition function $Z_0 = Z(T_0)$ at high
		temperature. 
		Next, for each $i$, we can estimate the ratio $\alpha_i$ by sampling from a distribution that is sufficiently close to the
		Boltzmann distribution $\pi_i$ \eqref{boltzmann} at temperature $T_i$ (see Section \ref{classicsection} for more detail). 
		This is possible by using a rapidly-mixing Markov chain $P_i$ whose stationary distribution is equal to the 
		Boltzmann distribution $\pi_i$.
		
		To be efficient, these classical schemes require that
		\begin{enumerate}
		\item we use a cooling schedule such that the resulting ratios $\alpha_i = Z(T_{i+1})/Z(T_i)$ are lower bounded by a constant $c^{-1}$ (to simplify the presentation, we use $c=2$ from now on), 
		\item the spectral gaps of the Markov chains $P_i$ are bounded from below by $\delta$.
		\end{enumerate}
		The time complexity of such FPRAS, i.e., the number of times we have to invoke an update step
		for a Markov chain from $\{P_1,\ldots,P_{\ell-1}\}$, is
		\begin{equation}
		         \tilde{O}\left(\frac{\ell^2}{\delta \cdot \epsilon^2}\right)\,,
		\end{equation}
		where $\tilde{O}$ means up to logarithmic factors.

		Our main result is a general method for `quantizing' such algorithms. 
		
		\begin{theorem}
		\label{mainresult}
    Consider a classical FPRAS for approximating the Gibbs partition function of a physical system
		at temperature $T_F$, satisfying the above conditions. Then, there exists a fully polynomial quantum approximation scheme that uses
		\begin{equation}
		   \tilde{O}\left(\frac{\ell^2}{\sqrt{\delta} \cdot \epsilon}\right)
		\end{equation}
		applications of a controlled version of a quantum walk operator from
		$\{W(P_1),\ldots,W(P_{\ell-1})\}$.
		\end{theorem}
		
The reduction in complexity for our quantum algorithm (in comparison to the
classical FPRAS) is twofold. First, we reduce the factor $1/\delta$ to $1/\sqrt{\delta}$ by using quantum
walks instead of classical Markov chains, and utilizing the quadratic relation
between spectral and phase gaps. As observed in \cite{Magniez},
this relation is at the heart of many quantum search algorithms based on quantum walks
(see e.g. \cite{Santha} for an overview of such quantum algorithms).
Second, we speed up the way to determine the ratios $\alpha_i$ by using the
quantum phase estimation algorithm in a novel way. This results in the reduction of the factor
$1/\epsilon^2$ to $1/\epsilon$.
		
The quantum algorithm we present builds upon our previous work \cite{WA:08}, where two of us have shown how to use quantum walks to approximately prepare coherent encodings
\begin{eqnarray}
	|\pi_i\> = \sum_{\sigma \in \Omega} \sqrt{\pi_i(\sigma)} \ket{\sigma}
\end{eqnarray}
of stationary 
distributions $\pi_i$ of Markov chains $P_i$, provided that the Markov chains are slowly-varying.  Recall that a sequence of Markov chains is called slowly-varying if the stationary distributions of two adjacent chains are sufficiently close to each other.  As we will see later, this condition is automatically satisfied for Markov chains that are used in FPRAS for approximating partition functions.
		
Note that our objective of approximately preparing coherent encodings of stationary distributions is different from the objective in \cite{Richter1}, where the author seeks to speed up the process of approximately preparing density operators encoding stationary distributions. For our purposes, we have to work with coherent encodings because otherwise we could not achieve the 
second reduction from $1/\epsilon^2$ to $1/\epsilon$.

The paper is organized as follows. In Section \ref{classicsection} we review
the classical FPRAS in more detail. 
We present our quantum algorithm in two steps. 
First, in Section \ref{perfect} we explain how our quantum algorithm works,  
assuming that we can perfectly and efficiently prepare coherent encodings of the distributions \eqref{boltzmann}.
Then, in Section \ref{imperfect} we describe the full quantum algorithm, 
dropping the assumption of Section \ref{perfect} 
and using approximate procedures for quantum sample preparation and readout, which are based on the quantum walks. 
We perform a detailed analysis of accumulation of error due to the approximation procedures 
and show that the success probability remains high, establishing Theorem~\ref{mainresult}.
Finally, in Section \ref{annealingsection} we conclude with a discussion
of open questions, the connection of our algorithm to simulated annealing,
and the directions for future research.
		
	
\section{Structure of the Classical Algorithm} \label{classicsection}		

Here we describe the classical approximation schemes in more detail, following closely the presentation in \cite[Section 2.1]{Vazirani}.
Choosing a sequence of temperatures
$T_0 \geq T_1 \geq \dots \geq T_\ell$ 
starting with $T_0=\infty$, and ending with
the desired final (low) temperature $T_\ell = T_F$,
we can express the Gibbs
partition function \eqref{partition} as a telescoping product \eqref{telescope}. 
At $T_0=\infty$, the partition function $Z_0$ is equal to
\begin{eqnarray}
   Z_0 = |\Omega|,
\end{eqnarray}
the size of the state space. On the other hand, for each $i=0,\dots,\ell-1$, we
can estimate the ratio
\begin{eqnarray}
	\alpha_i = \frac{Z_{i+1}}{Z_i}
	\label{alpharatio}
\end{eqnarray}
in \eqref{telescope} as follows.  Let $X_i \sim \pi_i$ denote a random state chosen according to the Boltzmann distribution $\pi_i$, i.e., 
\begin{equation}
	\Pr(X_i=\sigma)=\pi_i(\sigma)\,.
\end{equation}
Define a new random variable $Y_i$ by 
\begin{equation}
	Y_i = e^{-(\beta_{i+1}-\beta{i})\, E(X_i)},
\label{yvariable}
\end{equation}
where $\beta_i = (kT_i)^{-1}$ is the inverse temperature ($k$ is the Boltzmann constant). 
This $Y_i$ is an unbiased estimator for $\alpha_i$ since 
\begin{eqnarray}
	\EE{Y_i} 
	& = & \sum_{\sigma \in \Omega} \pi_i(\sigma) \, e^{-(\beta_{i+1}-\beta{i})\, E(\sigma)} \\
	& = &  
	\sum_{\sigma \in \Omega} \frac{e^{-\beta_i E(\sigma)}}{Z_i} \, e^{-(\beta_{i+1}-\beta{i})\, E(\sigma)} \\
	& = & 
	\sum_{\sigma \in \Omega} \frac{e^{-\beta_{i+1}\, E(\sigma)}}{Z_i}
	= 
	\frac{Z_{i+1}}{Z_i} = \alpha_i.
\label{expecty}
\end{eqnarray}
Assume now that we have an algorithm for generating states $X_i$ according to $\pi_i$. We draw 
\begin{eqnarray}
	m:=64\ell/\ep^2
\end{eqnarray}
samples of $X_i$ and take the mean $\overline{Y}_i$ of their corresponding estimators $Y_i$.
Then, the mean $\overline{Y}_i$ satisfies
\begin{eqnarray}
	\frac{\VV{\overline{Y}_i}}{\left(\EE{\overline{Y}_i}\right)^2} 
	=
	\frac{\ep^2}{64\ell}\, \frac{\VV{Y_i}}{\left(\EE{Y_i}\right)^2}
	\leq
	\frac{\ep^2}{16 \ell}\,.
\end{eqnarray}
(We have used the assumption $\frac{1}{2}\le\alpha_i\le 1$.)  We can now compose such estimates of $\alpha_i$. 
Define a new random variable $\overline{Y}$ by 
\begin{equation}
	\overline{Y} = \overline{Y}_{\ell-1}\overline{Y}_{\ell-2}\cdots\overline{Y}_0
\end{equation}
Since all $\overline{Y}_i$ are independent, we have
\begin{eqnarray}
	\EE{\overline{Y}}
	=
	\EE{Y_{\ell-1}} \EE{Y_{\ell-2}} \cdots \EE{Y_{0}}
	= 
	\alpha_{\ell-1} \alpha_{\ell-2} \cdots \alpha_{0} = \alpha,\nonumber
\end{eqnarray}
Moreover, $\overline{Y}$ has the property
\begin{eqnarray}
	\frac{\VV{\overline{Y}}}{\left(\EE{\overline{Y}}\right)^2}
	& = &
		\frac{\EE{\overline{Y}^2_{\ell-1}}\cdots\EE{\overline{Y}^2_{0}} -
				\EE{\overline{Y}_{\ell-1}}^2\cdots\EE{\overline{Y}_{0}}^2}		         
		{\EE{\overline{Y}^2_{\ell-1}}^2\cdots\EE{\overline{Y}_{0}}^2}
		\nonumber \\
	& = &
		         \left(
		                 1 +
		                 \frac{\VV{\overline{Y}_{\ell-1}}}{\left(\EE{\overline{Y}_{\ell-1}}\right)^2}
		         \right)
		         \cdots
		         \left(
		                 1 +
		                 \frac{\VV{\overline{Y}_{0}}}{\left(\EE{\overline{Y}_{0}}\right)^2}
		         \right)- 1 \nonumber\\
		         &\le&
		                 \left(e^{\ep^2/16\ell}\right)^\ell - 1 \label{use1}
		                 \nonumber\\ &\le&
		                 \epsilon^2/8\,, \label{use2} \nonumber
\end{eqnarray}
where we used $1+x\le e^x$ (true for all $x$) and $e^x-1\le 2x$ (true for all $x\in[0,1])$ in the last two steps, respectively.
Chebyshev's inequality now implies that the value of
$\overline{Y}$ is in the interval $[(1-\epsilon)\alpha, (1+\epsilon) \alpha]$
with probability at least $\frac{7}{8}$.
		
Of course, we are not able to obtain perfect samples $X_i$ from $\pi_i$.  Assume now that we have $X_i'$ that are from a distribution with a variation distance from $\pi_i$ smaller than 
\begin{eqnarray}
	d:= \epsilon^2/(512\ell^2).
\end{eqnarray}
Let $\overline{Y}'$ be defined as $\overline{Y}$ as above, but instead of $X_i$ we use $X_i'$.  Then, with probability at least $\frac{7}{8}$, we have $\overline{Y}=\overline{Y}'$.  To derive this, observe that the algorithm can be thought to first take a sample from a product probability distribution $\pi$ on the $(m\ell)$-fold direct product of $\Omega$.  We denote the probability distribution in the case of imperfect samples by $\pi'$.  The total variation distance between $\pi$ and $\pi'$ is then bounded from above by 
\begin{equation}
d \cdot m \cdot \ell = \frac{\epsilon^2}{512\ell^2} \cdot \frac{64\ell}{\epsilon^2} \cdot \ell = \frac{1}{8}\,.
\end{equation}
Therefore, $\overline{Y}'$ is in the interval $[(1-\epsilon)\EE{Y},(1+\epsilon)\EE{Y}]$ with probability at least $\frac{3}{4}$.

We obtain the samples $X'_i$ by applying Markov chains $P_i$ whose limiting distributions are equal to $\pi_i$.  
Constructing such rapidly-mixing Markov chains is a hard task, but it has been done for 
the Ising model \cite{Jerrum2} and the Potts model \cite{Vazirani}.
	
	
\section{Quantum Algorithm}

\subsection{Overview}
\label{overviewsection}

The classical FPRAS we described in Section \ref{classicsection} 
consists of 
\begin{enumerate}
	\item preparing many samples from a distribution close to $\pi_i$ 
		by letting a suitable Markov chain mix, 
	\item using these samples to approximate the ratios $\alpha_i$ in \eqref{telescope}, and
	\item composing these estimates of $\alpha_i$ into an estimate of the partition function.
\end{enumerate}	
We build our quantum algorithm on this scheme, with two novel quantum ingredients.
First, instead of letting a Markov chain $P_i$ mix towards its stationary distribution $\pi_i$,
we choose to approximately prepare the state $\ket{\pi_i} = \sum_{\sigma} \sqrt{\pi_i(\sigma)} \ket{\sigma}$,
a coherent encoding of the Boltzmann distribution. 
We use a preparation method \cite{WA:08} based on Grover's $\frac{\pi}{3}$-fixed-point search \cite{Grover},
efficiently driving the state $\ket{\pi_0}$ towards the desired state $|\pi_i\>$ 
through a sequence of intermediate states. 

Second, instead of using classical samples from the distribution $\pi_i$, 
we approximate $\alpha_i$ by phase-estimation of a certain unitary on the state $\ket{\pi_i}$. 
This is a new concept, going beyond our previous work \cite{WA:08}. 
This phase-estimation subroutine can be efficiently (albeit only approximately) 
applied by utilizing quantum walks.

\begin{figure}
\begin{center} 
	\medskip
	\hspace{3cm}
	\Qcircuit @C=1em @R=1em {
		\lstick{|\pi_0\>}                                                                                  
				& \gate{\cM_0} &&& \textrm{(obtain $\alpha_0$)} \\ 
		\lstick{|\pi_0\> \rightarrow |\tilde{\pi}_{1}\>}                                                   
				& \gate{\cM_1} &&& \textrm{(obtain $\alpha_1$)} \\ 
				&              &&& \\
		\lstick{|\pi_0\> \rightarrow |\tilde{\pi}_{1}\> \rightarrow \cdots \rightarrow |\tilde{\pi}_{\ell-1}\>} 
				& \gate{\cM_{\ell-1}} &&& \textrm{\quad(obtain $\alpha_{\ell-1}$)}   
	}
\end{center}
\caption{Structure of the quantum algorithm.}
\label{figstructure}
\end{figure}
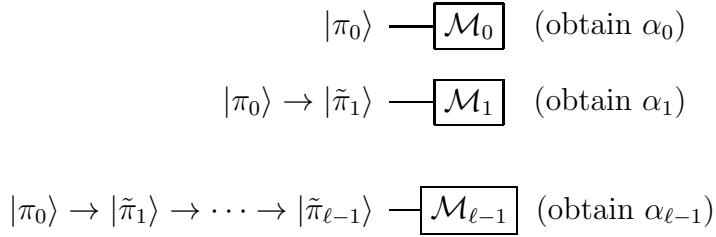

The structure of our algorithm is depicted in Fig. \ref{figstructure}.
It consists of successive approximate preparations of $\ket{\pi_i}$ followed by 
a quantum circuit outputting a good approximation to $\alpha_i$ (with high probability).
Our main result is the construction of a fast quantum version of a class of classical algorithms, 
summed in Theorem~\ref{mainresult}. 

We arrive at our quantum algorithm in two steps.
First, in Section \ref{perfect}, we explain how to quantize the the classical algorithm 
in the perfect case, assuming that we can take perfect samples $X_i$ from $\pi_i$.
Then, in Section \ref{imperfect} we release this assumption and describe the full quantum algorithm.


\subsection{Perfect Case}\label{perfect}

To estimate the ratios $\alpha_i$ in \eqref{telescope}, the classical algorithm generates 
random states $X_i$ from $\pi_i$ and computes the mean $\overline{Y}_i$ of the random variables $Y_i$.
The process of generating a random state $X_i$ from $\pi_i$ is equivalent to preparing the mixed state
\begin{equation}
	\rho_i = \sum_{\sigma\in\Omega} \pi_i(\sigma) |\sigma\>\<\sigma|\,.
\end{equation}
Instead of this, we choose to prepare the pure states
\begin{equation}
	|\pi_i\> = \sum_{\sigma\in\Omega} \sqrt{\pi_i(\sigma)} |\sigma\>\,.
\end{equation}
We call these states {\em quantum samples} since they coherently encode the probability 
distributions $\pi_i$. In this Section, we assume that we can prepare these exactly and efficiently.

The random variable $Y_i$ can be interpreted as the outcome of the measurement of the observable
\begin{equation}
	A_i = \sum_{\sigma\in\Omega} y_i(\sigma) |\sigma\>\<\sigma|
\end{equation}
in the state $\rho_i$, where 
\begin{equation}
y_i(\sigma) = e^{-(\beta_{i+1}-\beta_i) E(\sigma)}\,.
\end{equation}
With this interpretation in mind, we see that to estimate $\alpha_i$ 
classically, we need to estimate the expected value ${\rm Tr}(A_i \rho_i)$ by repeating the above measurement 
several times and outputting the mean of the outcomes.

We now explain how to quantize this process.  We add an ancilla qubit to our quantum system in which the quantum samples $|\pi_i\>$ live.  For each $i=0,\ldots,\ell-1$, we define the unitary
\begin{equation}
V_i = \sum_{\sigma\in\Omega} \ket{\sigma}\bra{\sigma} \otimes
\left(
\begin{array}{cc}
  \sqrt{y_i(\sigma)}   & \sqrt{1-y_i(\sigma)} \\
- \sqrt{1-y_i(\sigma)} & \sqrt{y_i(\sigma)}
\end{array}
\right)\,.
\end{equation}
This $V_i$ can be efficiently implemented, it is a rotation on the extra qubit controlled
by the state of the first tensor component.
Let us label
\begin{equation}
	\ket{\psi_i} = V_i \big( |\pi_i\> \otimes |0\> \big).
\label{psistate}
\end{equation}
Consider now the expected value of the projector
\begin{equation}
	P = \ii \otimes \ket{0}\bra{0}
\label{projector}
\end{equation}
in the state $\ket{\psi_i}$. We find
\begin{equation}
	\<\psi_i|P|\psi_i\> = \<\pi_i|A_i|\pi_i\> = \alpha_i\,.
\end{equation}
We now show how to speed up the process of estimating $\alpha_i$ 
with a method that generalizes quantum counting \cite{BrassardHoyerTapp}. 
As noted in the beginning of this Section, we assume efficient preparation of $\ket{\pi_i}$, which in turn
implies that we can efficiently implement the reflections
\begin{equation}\label{eq:reflectPii}
	R_i = 2|\pi_i\>\<\pi_i| - \ii\,.
\end{equation}
The result of this Section, the existence of a quantum FPRAS for estimating the partition function 
assuming efficient and perfect preparation of $\ket{\pi_i}$, is summed in Theorem \ref{th:perfectZ}:
\begin{theorem}\label{th:perfectZ}
	There is a fully polynomial quantum approximation scheme $\cA$ for the partition function $Z$.  
	Its output $Q$ 	satisfies
	\begin{equation}
		\Pr\big[(1-\epsilon) Z \leq Q \leq (1+\epsilon) Z \big] \geq \frac{3}{4}\,.  
	\end{equation}
	For each $i=0,\ldots,\ell-1$, the scheme $\cA$ uses 
	\begin{equation}
		O\left(\log \ell \right)
	\end{equation}
	perfectly prepared quantum samples $|\pi_i\>$, and applies the controlled-$R_i$ operator
	\begin{equation}
		O\left( \frac{\ell}{\ep} \log \ell \right)
		\end{equation}
		times, where $R_i$ is as in (\ref{eq:reflectPii}).
\end{theorem}
	   
To prove Theorem~\ref{th:perfectZ}, we need the following three technical results.

\begin{lemma}[Quantum ratio estimation]\label{lem:generalPE}
Let $\ep_{pe}\in (0,1)$. For each $i=0,\ldots,\ell-1$ there exists a quantum approximation scheme $\cA'_i$ for $\alpha_i$. Its output $Q'_i$ satisfies
\begin{equation}
\Pr\big[(1-\ep_{pe}) \alpha_i \leq Q'_i \leq (1+\ep_{pe}) \alpha_i \big] \geq \frac{7}{8}.
\end{equation}
The scheme $\cA'_i$ requires one copy of the quantum sample $|\pi_i\>$ and invokes the controlled-$R_i$ operator
$O\left( \ep^{-1}_{pe}\right)$ times, where $R_i$ is as in (\ref{eq:reflectPii}).
\end{lemma}
		
		\begin{proof}
		Let
		\begin{equation}
			G = (2 \ket{\psi_i}\bra{\psi_i} - \ii)\,(2P - \ii)\,.
		\end{equation}
		Define the basis states
		\begin{equation}\label{eqn:basis}
		         \ket{\gamma_1} = \frac{(\ii - P) \ket{\psi_i}}{\sqrt{1-\alpha_i}}\,,
		         \quad\mbox{and} \quad
		         \ket{\gamma_2} = \frac{P \ket{\psi_i}}{\sqrt{\alpha_i}}\,.
		\end{equation}
		Restricted to the plane spanned by $|\gamma_1\>$ and $|\gamma_2\>$, $G$ acts as a rotation
		\begin{eqnarray}
		         G\big|_{\{\ket{\gamma_1},\ket{\gamma_2}\}} =
		         \left(
		                 \begin{array}{rr}
		                     \cos\theta & \sin\theta \\
		                   - \sin\theta & \cos\theta
		        		 \end{array}
		         \right)\,,
		         \label{gunitary}
		\end{eqnarray}
		where $\theta\in [0,\frac{\pi}{2}]$ satisfies
		\begin{eqnarray}
		         \cos\theta  = 2 \alpha_i - 1.
		         \label{pfromtheta}
		\end{eqnarray}
		The eigenvectors and eigenvalues of $G$ are
		\begin{equation}
		         \ket{G_{\pm}} = \frac{1}{\sqrt{2}} \left[\begin{array}{r}1\\ \pm i\end{array}\right]\,, \quad
		         \lambda_{\pm} = e^{\pm i \theta}\,.
		\end{equation}
		We do not have direct access to one of these eigenvectors, as the state
		$\ket{\psi_i}$ is in a superposition of $\ket{G_{+}}$ and $\ket{G_{-}}$. Thus, when we apply the phase estimation
		circuit for the unitary $G$ to the state $\ket{\psi_i}$, we will sometimes obtain an estimate of
		$\theta$, and sometimes an estimate of $2\pi - \theta$.  However, this is not a problem since both 
		$\theta$ and $2\pi-\theta$ plugged into \eqref{pfromtheta} yield the same result for $\alpha_i$.
		
		We require that the estimate $\theta'$ satisfies
		\begin{eqnarray}
		  \left| \theta' - \theta \right| \leq 2 \ep_{pe} \, \alpha_i \leq \ep_{pe} \label{thetaprecision}
		\end{eqnarray}
		with probability at least $\frac{7}{8}$.  Using the phase estimation circuit in \cite{NielsenChuang}, this means that 
		$\frac{\theta'}{2 \pi}$ has to be an $n_a = \log \frac{2\pi}{\ep_{pe}}$ bit approximation of 
		the phase and the failure probability $p_f$ has to be less than $\frac{1}{8}$. 
		To achieve this, it suffices to use a phase estimation circuit (see Fig.~\ref{figphase}) with
		\begin{equation}
		  t = \log \frac{2\pi}{\ep_{pe}} + \log \left( 2 + \frac{1}{2\,p_f}\right) 
		    =  O\left(\log \ep_{pe}^{-1}\right) \nonumber
		\end{equation}
		ancilla qubits.  This circuit invokes the controlled-$G$ operation
		$O(2^t) = O\left( \ep^{-1}_{pe} \right)$ times. 
		\begin{figure}
		\begin{center}
		         \hspace{1cm}
		         \Qcircuit @C=0.6em @R=0.6em {
		         \lstick{|0\>}    & \qw & \gate{H} & \qw            & \qw             &
		\qw &
		                         & & \qw & \ctrl{4} & \qw      & \multigate{3}{{\rm
		DFT}^\dagger} & \qw                                      \\
		                          &     &          & \vdots         &                 &
		& \cdots
		      & &     &                     &     &                            &
		\\
		         \lstick{|0\>}    & \qw & \gate{H} & \qw            & \ctrl{2}        &
		\qw &
		                         & & \qw & \qw                 & \qw & \ghost{{\rm
		DFT}^\dagger}  & \qw                                     \\
		         \lstick{|0\>}    & \qw & \gate{H} & \ctrl{1}       & \qw             &
		\qw &
		                         & & \qw & \qw                 & \qw & \ghost{{\rm
		DFT}^\dagger}  & \qw                                     \\
		         \lstick{|\psi\>} & \qw & \qw      & \gate{G^{2^0}} & \gate{G^{2^1}}  &
		\qw &
		                         & & \qw & \gate{G^{2^{t-1}}}  & \qw & \qw
		& \qw
		         }
		\end{center}
		\caption{A basic phase estimation circuit with $t$ ancilla qubits}
		\label{figphase}
		\end{figure}
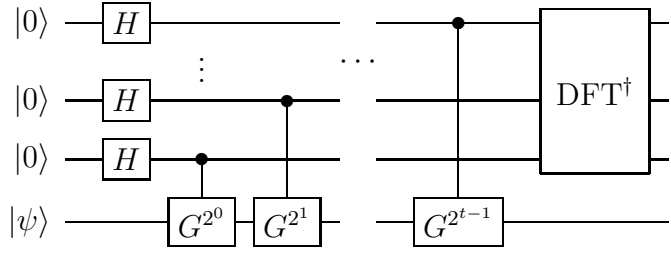		
		
		Let $\alpha'_i$ denote the value we compute from the estimate $\theta'$.  We have
		\begin{equation}
		   \left| \alpha_i - \alpha'_i \right| = \frac{1}{2} \left| \cos\theta - \cos\theta' \right| \leq
		   \frac{1}{2} \left| \theta - \theta' \right| \leq \ep_{pe} \, \alpha_i \,,
		\end{equation}
		showing that the estimate $\alpha'_i$ is within $\pm\epsilon_{pe} \alpha_i$ of the exact value $\alpha_i$ with probability at least $\frac{7}{8}$. 
		This completes the proof that the random variable $Q'_i$ corresponding to the output satisfies the desired properties on estimation accuracy and success probability.
		\end{proof}

		We can boost the success probability of the above quantum approximation scheme for the ratio $\alpha_i$ 
		by applying the \emph{powering lemma} from
		\cite{Valiant}, which we state here for completeness:
		
\begin{lemma}[Powering lemma for approximation schemes]\label{lem:powering}
Let $\cB'$ be a (classical or quantum) approximation scheme whose estimate $W'$ is within $\pm\ep_{pe} q$ to some value $q$ with probability $\frac{1}{2}+\Omega(1)$.  Then, there is an approximation scheme $\cB$ whose estimate $W$ satisfies
\begin{equation}
\Pr\big[(1-\ep_{pe}) q \leq W \leq (1+\ep_{pe}) q \big] \geq 1- \delta_{boost}\,.
\end{equation}
It invokes the scheme $\cB'$ as a subroutine $O\left(\log \delta_{boost}^{-1} \right)$ times.
\end{lemma}

With the help of Lemma \ref{lem:powering}, we now have the constituents required to compose the individual 
estimates of $\alpha_i$ into an approximation for the partition function \eqref{telescope}.
		
\begin{lemma}\label{lem:boundProdEstimates}
Let $\epsilon>0$. Assume we have approximation schemes $\cA_0,\cA_1,\ldots,\cA_{\ell-1}$ such that their estimates $Q_0,Q_1,\ldots,Q_{\ell-1}$ satisfy
\begin{equation}
\Pr\left[ \Big(1-\frac{\ep}{2\ell}\Big) \alpha_i \le Q_i \le \Big(1+\frac{\ep}{2\ell}\Big) \alpha_i \right] \ge 1 - \frac{1}{4\ell}\,.
\end{equation}
Then, there is a simple approximation scheme $\cA$ for the product $\alpha=\alpha_0 \alpha_1 \cdots \alpha_{\ell-1}$. 
The result $Q=Q_0 Q_1 \cdots Q_{\ell-1}$ satisfies
\begin{equation}\label{eq:multApprox}
\Pr\big[ (1-\ep) \alpha \le Q \le (1+\ep) \alpha \big] \ge \frac{3}{4}\,.
\end{equation}
\end{lemma}
		
		\begin{proof} 
		For each $i = 0,\dots,\ell-1$, the failure probability for estimating $\alpha_i$ is smaller than $1/(4\ell)$.
		The union bound implies that the overall failure probability is smaller than $1/4$,
		proving the lower bound $\frac{3}{4}$ on the success probability in \eqref{eq:multApprox}.
		
		To obtain the upper bound on the deviation, we now assume that each $Q_i$ takes
		the upper bound value. We have
		\begin{eqnarray}
		\frac{Q - \alpha}{\alpha}
		& \leq &
		\prod_{i=0}^{\ell-1} \left(1 + \frac{\ep}{2\ell}\right) - 1
		=
		\left(1 + \frac{\ep}{2 \ell}\right)^{\ell} - 1 \nonumber \\
		& \le &
		e^{\epsilon/2} - 1 \le \epsilon\,, \nonumber 
		\end{eqnarray}
		where we have used $1+x\le e^x\le 1 + 2x$, which is true for all $x\in [0,1]$. Thus, in the case of success, we have $Q\leq(1+\ep) \alpha$.
		
		To obtain the lower bound on the deviation, we assume that each $Q_i$ takes its
		lower bound value. We have
		\begin{equation}
		\frac{\alpha - Q}{\alpha}
		\leq
		1 - \prod_{i=0}^{\ell-1} \left(1 - \frac{\ep}{2\ell}\right)
		\leq
		\sum_{i=0}^{\ell-1} \frac{\ep}{2\ell} \le \epsilon\,,
		\end{equation}
		where we have used $\left|\prod_i x_i - \prod_i y_i\right| \leq
		\sum_{i}\left|x_i - y_i\right|$,
		true for arbitrary $x_i,y_i\in [0,1]$. Thus, in the case of success, we have $(1-\ep)\alpha\leq Q$.
		\end{proof}

We are now ready to prove Theorem \ref{th:perfectZ}:
\begin{proof}[Proof of Theorem \ref{th:perfectZ}]

For each $i=0,\ldots,\ell-1$, we can apply Lemma~\ref{lem:generalPE} with the state $|\psi_i\>$ (\ref{psistate}) 
and the projector $P$ (\ref{projector}). This gives us a quantum approximation scheme for $\alpha_i$.
Note that to prepare $|\psi_i\>$, it suffices to prepare $|\pi_i\>$ once. Also, to realize a controlled reflection 
around $|\psi_i\>$, it suffices to invoke the controlled reflection around $|\pi_i\>$ once. 

We now use the reflection $2|\psi_i\>\<\psi_i|-\ii$ and set $\epsilon_{pe}=\epsilon/(2\ell)$ in Lemma~\ref{lem:generalPE}.
With these settings, we can apply Lemma \ref{lem:powering} to the resulting approximation scheme 
for $\alpha_i$ with $\delta_{boost}=1/(4\ell)$. This gives us approximation schemes $\cA_i$ outputting $Q_i$ with 
high precision and probability of success that can be used in Lemma~\ref{lem:boundProdEstimates}. 
The composite result $Q=Q_0 \cdots Q_{\ell-1}$ is thus an approximation 
for $\alpha=\alpha_0 \cdots \alpha_{\ell-1}$ 
with the property
\begin{eqnarray}
	\Pr\big[ (1-\ep) \alpha \le Q \le (1+\ep) \alpha \big] \ge \frac{3}{4}\,.
\end{eqnarray}
Finally, we obtain the estimate for $Z$ by multiplying $Q$ with $Z_0$.
Let us summarize the costs from Lemmas \ref{lem:generalPE}-\ref{lem:boundProdEstimates}.
For each $i=0,\dots,\ell-1$, this scheme uses $\log \delta_{boost}^{-1} = O(\log \ell)$ copies of the state $\ket{\pi_i}$, and invokes $\left(\log \delta_{boost}^{-1}\right) \ep_{pe}^{-1} = O\left(\frac{\ell}{\ep} \log \ell\right)$
reflections around $\ket{\pi_i}$.
\end{proof}		
		
		
\subsection{Quantum FPRAS}\label{imperfect}
		
In the previous Section we have assumed that we can prepare the quantum samples $|\pi_i\>$ 
and implement the controlled reflections $R_i=2|\pi_i\>\<\pi_i|-\ii$ about these states 
perfectly and efficiently. We now release these assumptions and show how to approximately accomplish both tasks 
with the help of quantum walks operators.  We then show that the errors arising from these approximate procedures 
do not significantly decrease the success probability of the algorithm.
This will wrap up the proof of our main result, Theorem~\ref{mainresult}.

In \cite{WA:08}, two of us how to approximately prepare quantum samples $|\pi_i\>$ 
of stationary distributions of slowly-varying Markov chains. Using the fact that the consecutive
states $|\pi_i\>$ and $|\pi_{i+1}\>$ are close, we utilize Grover's $\frac{\pi}{3}$ fixed-point search \cite{Grover} to drive the starting state $\ket{\pi_0}$ towards the desired state $|\pi_i\>$ through multiple intermediate steps. Moreover, to be able to perform this kind of Grover search, we have to be able to apply selective phase shifts of the form $S_i = \omega |\pi_i\>\<\pi_i| + (\ii - |\pi_i\>\<\pi_i|)$ for $\omega=e^{i\pi/3}$ and $\omega=e^{-i\pi/3}$. 
This is another assumption of Section \ref{perfect} that we have to drop here. Nevertheless, an efficient way to apply these phase shifts  approximately, based on quantum walks and phase estimation, exists \cite{WA:08}. 

Our task is to show that the approximation scheme from Lemma~\ref{lem:generalPE} 
works even with approximate input states and using only approximate reflections 
about the states $\ket{\pi_i}$. Let us start with addressing the approximate state preparation.
To be able to use the results of \cite{WA:08}, we first have to establish an important condition.
For their method to be efficient, the overlap of two consecutive quantum samples $|\pi_i\>$ and $|\pi_{i+1}\>$ has to be large. This is satisfied when $\alpha_i=Z_{i+1}/Z_i$ is bounded from below by $\frac{1}{2}$, since
\begin{eqnarray}
	|\<\pi_i|\pi_{i+1}\>|^2
	& = &
	\left|
	\sum_{\sigma\in\Omega} \frac{\sqrt{e^{-\beta_i E(\sigma)} \, 
	e^{-\beta_{i+1} E(\sigma)}}}{\sqrt{Z_i \, Z_{i+1}}} \right|^2 \nonumber \\
	& \ge &
	\left|
	\frac{\sum_{\sigma\in\Omega} \, e^{-\beta_{i+1}E(\sigma)}}{\sqrt{2 Z_{i+1}} \, \sqrt{Z_{i+1}}} \right|^2 
	= 
	\frac{1}{2}\,. \nonumber
\end{eqnarray}
The following lemma then directly follows from the arguments used in \cite[Theorem~2]{WA:08}.
\begin{lemma}\label{samplinglemma}
For $\ep_S>0$ arbitrary and each $i=1,\ldots,\ell-1$, there is a quantum method preparing a state $|\tilde{\pi}_i\>$ with
\begin{eqnarray}
\| |\tilde{\pi}_i\> - |\pi_i\>|0\>^{\otimes a} \| \leq {\ep_S} \,,
\end{eqnarray}
where $a=O\left(\frac{\ell}{\ep_S \sqrt{\delta}}\right)$ is the number of ancilla qubits.  The method invokes a controlled version of a walk operator from the set $\{W(P_1), \ldots, W(P_{\ell-1})\}$ 
\begin{eqnarray}
O\left( \frac{\ell}{\sqrt{\delta}} \log^2 \frac{\ell}{\ep_S} \right)\,.
\end{eqnarray}
times.
\end{lemma}
		
We choose the preparation method from Lemma \ref{samplinglemma} with $\ep_S = \frac{1}{32}$. 
The cost for this precision $\ep_S$ is
\begin{equation}
	O\left( \frac{\ell}{\sqrt{\delta}} \log^2 \ell \right)\,
\label{samplecost}
\end{equation}
applications of the quantum walk operator.
Recall that when we used Lemma~\ref{lem:generalPE} in Section \ref{perfect} with the state $\ket{\psi_i}$ (coming from the perfect quantum sample $\ket{\pi_i}$) as input, the success probability of the resulting scheme was
greater than $\frac{7}{8}$. We now use the method given in Lemma~\ref{lem:generalPE} on the approximate input $|\tilde{\psi}_i\>= V_i (|\tilde{\pi}_i\> \otimes |0\>)$. With our chosen precision for preparing $|\tilde{\pi}_i\>$, the success probability of the approximation scheme of Lemma~\ref{lem:generalPE} cannot decrease by more than $2\cdot \frac{1}{32}$.

	
The second assumption of Lemma~\ref{lem:generalPE} we need to drop is the ability to perfectly implement the reflections $R_i=2|\pi_i\>\<\pi_i|-\ii$. 
We now show how to approximately implement these reflections. The following lemma follows directly from the arguments in \cite[Lemma~2 and Corollary~2]{WA:08}.
\begin{lemma}\label{reflectlemma}
For $\epsilon_R>0$ arbitrary and each $i=1,\ldots,\ell-1$, there is an approximate reflection $\tilde{R}_i$ such that 
\begin{equation}
\tilde{R}_i \Big(  |\varphi\> \otimes |0\>^{\otimes b}\Big) = \Big(R_i|\varphi\>\Big) \otimes |0\>^{\otimes b} + |\xi\>
\end{equation}
where $|\varphi\>$ is an arbitrary state, $b=O\left( \log \epsilon_R^{-1} \,\log \frac{1}{\sqrt{\delta}}\right)$ is the number of ancilla qubits, and $|\xi\>$ is some error vector with $\||\xi\>\|\le \epsilon_R$.
It invokes the controlled version of a walk operator from $\{W(P_1),\ldots,W(P_{\ell-1})\}$ 
\begin{equation}
O\left(\frac{1}{\sqrt{\delta}} \log\frac{1}{\epsilon_R}\right)
\end{equation}
times.
\end{lemma}
		
Recall that in Lemma~\ref{lem:generalPE}, the controlled reflection $R_i$ is invoked $O(1/\epsilon_{pe})$ times.  We now run this approximation scheme with $\tilde{R}_i$ instead of $R_i$. The norm of the accumulated error vector is
\begin{equation}
	O\left(\frac{1}{\ep_{pe}}\right) \cdot \ep_R.
\end{equation}
We choose
\begin{equation}
	\ep_R = \Omega(\ep_{pe})
\end{equation}
to bound the norm of the accumulated error from above by $\frac{1}{32}$. The success probability can then decrease by at most $2\cdot \frac{1}{32}$.

Combining these arguments establishes a variant of Lemma~\ref{lem:generalPE} without the 
unnecessary assumptions of Section \ref{perfect}:  
\begin{lemma}
\label{lem:approxscheme}
Let $\ep_{pe}\in (0,1)$. For each $i=0,\ldots,\ell-1$, there exists a quantum approximation scheme $\cA''_i$ for $\alpha_i$. Its estimate $Q''_i$ satisfies
\begin{equation}
	\Pr\big[(1-\ep_{pe}) \alpha_i \leq Q''_i \leq (1+\ep_{pe}) \alpha_i \big] \geq \frac{3}{4}.
\end{equation}
This scheme invokes the controlled version of a walk operator from $\{W(P_1),\ldots,W_{\ell-1}\}$ 
\begin{equation}\label{eq:costApproxA}
		O\left(
	\frac{\ell}{\sqrt{\delta}} \, \log^2 \ell 
	+ \frac{1}{\ep_{pe} \sqrt{\delta}}\, \log \ep_{pe}^{-1}
	\right)\,.
\end{equation}
\end{lemma}
\begin{proof}
The success probability of the scheme in Lemma \ref{lem:generalPE} was greater than $\frac{7}{8}$. 
Both the approximate state preparation and using approximate reflections reduce the overall 
probability of success by at most $\frac{1}{16}$. Thus the probability of success of the method
given in Lemma~\ref{lem:generalPE} after dropping the unnecessary assumptions is at least $\frac{3}{4}$.
\end{proof}

We can finally complete the proof of Theorem \ref{mainresult} by following the procedure that led to the proof of Theorem \ref{th:perfectZ} in Section \ref{perfect}. 
\begin{proof}[Proof of Theorem \ref{th:perfectZ}]
For each $i=0,\ldots,\ell-1$, we proceed as follows.
We use the approximation scheme $\cA''_{i}$ from Lemma~\ref{lem:approxscheme} with precision $\epsilon_{pe}=\ep/(2\ell)$.
We then boost the success probability of each $\cA_i''$ to $1-\frac{1}{4\ell}$ by applying the powering lemma (Lemma~\ref{lem:powering}) with $\delta_{boost}=1/(4\ell)$. This step increases the cost in (\ref{eq:costApproxA}) by the factor $O(\log \ell)$. This resulting scheme $\cA_i$ now satisfies the properties required for Lemma~\ref{lem:boundProdEstimates}. 
We can thus use it to obtain a composite approximation scheme whose output satisfies 
\begin{equation}
	\Pr\big[(1-\epsilon) Z \leq Q \leq (1+\epsilon) Z \big] \geq \frac{3}{4}\,.  
\end{equation}
The resulting cost of this scheme (the number of times we have to invoke the controlled 
quantum walk operators) is
\begin{equation}
	O\left(
		\frac{\ell^2}{\sqrt{\delta}} \, \log^3 \ell 
		+ 
		\frac{\ell^2}{\ep\sqrt{\delta}} 
		(\log\ell) 
		\left(
		\log \ell 
		+
		\log \ep^{-1} 
		\right)
	\right)
	=
	\tilde{O}
		\left(
			\frac{\ell^2}{\ep \sqrt{\delta}}
		\right)\,.
\end{equation}
\end{proof}


\section{Conclusions}\label{annealingsection}
		
We have shown how to quantumly speed up a classical FPRAS for approximating
partition functions, as measured in the number of times we have to invoke
a step of a quantum walk (instead of classical Markov chains). 
We obtained two reductions in complexity: $1/\delta \rightarrow 1/\sqrt{\delta}$ and 
$1/\epsilon^2 \rightarrow 1/\ep$.  
These two reductions are intimately related; they cannot occur separately.
If we used quantum samples merely to obtain classical samples (i.e., if we tried to estimate the ratios without phase estimation), then this would lead to $O(\ell^3)$ dependence (for $\ep\propto \ell^{-1}$).  This is because we would have to take $O(\frac{\ell}{\epsilon^2})$ classical samples for each $i$ and producing a quantum sample costs at least $O(\ell)$. The advantage of our approximation procedure based on quantum phase estimation is that it requires only one quantum sample (or more precisely, $\log \ell$, after using the powering lemma to boost the success probability).  We cannot obtain the second speed-up without using quantum samples (as mentioned in the introduction, this prevents us from using a procedure such as \cite{Richter1} that prepares density operators encoding stationary distributions).
Also, the arguments employed in the error analysis in the quantum case are quite different from those in the classical error analysis.

Each classical FPRAS we speed up uses the telescoping trick \eqref{telescope},
a particular cooling schedule (decreasing sequence of temperatures),
and slowly-varying Markov chains which mix rapidly, with stationary distributions 
equal to the Boltzmann distributions at the intermediate temperatures.
The classical FPRAS is useful only when we have the Markov chains with the required properties. Moreover, the cooling schedules need to be such that the ratios $\alpha_i$ \eqref{alpharatio}
are lower bounded by some $c^{-1}$.
In \cite{StefankovicVempalaVigoda}, the authors show that it is possible to use a cooling schedule $T_0=\infty > T'_1 > \ldots > T'_{\ell'-1} = T_F$ 
for estimating the partition function $Z(T_F)$ as long as for each $i$,
\begin{equation}\label{eq:Chebyshev}
\frac{\EE{Y_i^2}}{\big(\EE{Y_i}\big)^2} \le b,
\end{equation}
where $b$ is some constant.  Such a cooling schedule is called a Chebyshev cooling schedule.  Note that the above condition is automatically satisfied in the situation we consider in this paper, but not vice versa (recall that we assume that we have a cooling schedule such that $\EE{Y_i}$ is bounded from below by a constant for each $i$; we set it to $\frac{1}{2}$ for simplicity of presentation).  The advantage of Chebyshev cooling schedules is that they are provably shorter.  The authors present an adaptive algorithm for constructing Chebyshev cooling schedule.  We plan to explore if it is possible to speed up this process.  But even if this is possible, a potential obstacle remains.  It is not clear whether we can still obtain the reduction from $\frac{1}{\epsilon^2}$ to $\frac{1}{\epsilon}$ when we only know that the condition (\ref{eq:Chebyshev}) is satisfied.  It seems that the condition $\EE{Y_i}>c^{-1}$ with $c$ some constant is absolutely necessary for phase estimation to yield the quadratic speed-up with respect to the accuracy parameter $\epsilon$.

The combination of simulated annealing and the Monte Carlo Markov Chain method used in approximating partition functions is the central piece of the best currently known algorithm for estimating permanents with non-negative entries \cite{Vazirani}.  We therefore plan to explore where our techniques can be used to speed up this breakthrough classical algorithm.
	
	
	\section{Acknowledgments}
		A.~A., C.~C. and P.~W. gratefully acknowledge the support by NSF grants
		CCF-0726771 and	CCF-0746600. D.~N. gratefully acknowledges support by
		European Project QAP 2004-IST-FETPI-15848	and by the Slovak Research 
		and Development Agency under the contract No. APVV-0673-07.
		

\appendix

\section{Quantum Walks from Classical Markov Chains}
\label{walksection}

The class of classical approximation schemes that we speed up uses reversible, 
ergodic Markov chains $P_i$ with stationary distributions $\pi_i$. 
Here we briefly review the quantum analogue of a Markov chain, 
describing the {\em quantum walk operator} $W$ corresponding to the classical Markov Chain $P$.

In each step of a Markov chain $P$ with state space $\Omega$, the probability of a state $x$ to transition to 
another state $y$ is given by the element $p_{xy}$ of the $D\times D$ transition matrix, 
where $D=|\Omega|$. Following Szegedy \cite{Szegedy}, for each such Markov Chain, 
we can define its quantum analogue. 
The Hilbert space on which this quantum operation acts is $\C^{D}\otimes\C^{D}$, with two $\C^{D}$ registers.
We start by defining the states
\begin{eqnarray}
	\ket{p_x} = \sum_{y\in \Omega} \sqrt{p_{xy}}\ket{y}.
\end{eqnarray}
These states can be generated by a {\em quantum update} -- any unitary $U$ that satisfies
\begin{eqnarray}
	U \ket{x}\ket{0} = \ket{x} \ket{p_x}
\end{eqnarray}
for some fixed state $0\in \Omega$ and all $x\in \Omega$.
The quantum analogue of a Markov chain is then defined as follows.
\begin{definition}[Quantum Walk]
	A quantum walk $W(P)$ based on a classical reversible Markov chain $P$ is a unitary operation
	acting on the space $\C^D\otimes \C^D$ as
	\begin{eqnarray}
		W(P) = R_\cB \cdot R_\cA,
	\end{eqnarray}
	where $R_\cB$ and $R_\cA$ are reflections about the spaces
	\begin{eqnarray}
		\cA &=& \mathrm{span} \{ \ket{x}\ket{0} : x\in \Omega \}, \\
		\cB &=& U^\dagger S U \cA,
	\end{eqnarray}
	and $S$ is a swap of the two registers.
\end{definition}
This particular definition of the quantum walk is suitable for making some of the proofs
in \cite{WA:08} easier. It is equivalent to the standard definition of Szegedy \cite{Szegedy}
up to conjugation by $U$. Therefore, the spectral properties of our $W$ and Szegedy's quantum walk
are the same. 

Let $\delta$ be the spectral gap of the classical Markov chain $P$. Let us write
its eigenvalues as $\mu_0 = 1$ and $\mu_j = \cos(\theta_j)$, for $j=1,\dots,D-1$ and 
$\theta_j \in \left(0,\frac{\pi}{2}\right)$. 
According to Szegedy \cite{Szegedy}, on the space $\cA + \cB$, the eigenvalues of the quantum walk $W(P)$ 
with nonzero imaginary part are $e^{\pm 2i\theta_j}$. The phase gap of the quantum walk $W(P)$ is then defined
as $\Delta = 2\theta_1$ (with $\theta_1$ the smallest of $\theta_j$). When the Markov chain is ergodic
and reversible, Szegedy proved that
\begin{eqnarray}
	\Delta \geq 2\sqrt{\delta},
\end{eqnarray}
a quadratic relation between the phase gap $\Delta$ of the quantum walk $W(P)$
and the spectral gap $\delta$ of the classical Markov chain $P$. This quadratic relation is
behind the speed-up of many of today's quantum walk algorithms. 


\end{document}